\documentclass[sigconf]{acmart}

\usepackage{enumitem}
\usepackage{graphicx}
\usepackage{subfigure}
\usepackage{amsfonts}
\usepackage{amsmath}
\usepackage{makecell}
\usepackage{booktabs}
\usepackage{multirow}
\usepackage[normalem]{ulem}
\useunder{\uline}{\ul}{}
\usepackage{balance}
\usepackage[ruled]{algorithm2e} 

\AtBeginDocument{%
  }

\setcopyright{acmlicensed}
\copyrightyear{2024}
\acmYear{2024}
\acmDOI{XXXXXXX.XXXXXXX}

\acmConference[Conference acronym 'XX]{Make sure to enter the correct
  conference title from your rights confirmation emai}{June 03--05,
  2018}{Woodstock, NY}

\acmPrice{15.00}
\acmISBN{978-1-4503-XXXX-X/18/06}

\begin{document}

\title{Matryoshka Representation Learning for Recommendation}


\author{Riwei Lai}
\affiliation{
  \institution{Hong Kong Baptist University}
  \city{}
  \country{}
}
\email{csrwlai@comp.hkbu.edu.hk}

\author{Li Chen}
\affiliation{
  \institution{Hong Kong Baptist University}
  \city{}
  \country{}
}
\email{lichen@comp.hkbu.edu.hk}

\author{Weixin Chen}
\affiliation{
  \institution{Hong Kong Baptist University}
  \city{}
  \country{}
}
\email{cswxchen@comp.hkbu.edu.hk}

\author{Rui Chen}
\affiliation{
  \institution{Harbin Engineering University}
  \city{}
  \country{}
}
\email{ruichen@hrbeu.edu.cn}

\begin{abstract}
Representation learning is essential for deep-neural-network-based recommender systems to capture user preferences and item features within fixed-dimensional user and item vectors. Unlike existing representation learning methods that either treat each user preference and item feature uniformly or categorize them into discrete clusters, we argue that in the real world, user preferences and item features are naturally expressed and organized in a hierarchical manner, leading to a new direction for representation learning. In this paper, we introduce a novel \textbf{matryoshka representation learning method for recommendation} (MRL4Rec), by which we restructure user and item vectors into matryoshka representations with incrementally dimensional and overlapping vector spaces to explicitly represent user preferences and item features at different hierarchical levels. We theoretically establish that constructing training triplets specific to each level is pivotal in guaranteeing accurate matryoshka representation learning. Subsequently, we propose the matryoshka negative sampling mechanism to construct training triplets, which further ensures the effectiveness of the matryoshka representation learning in capturing hierarchical user preferences and item features. The experiments demonstrate that MRL4Rec can consistently and substantially outperform a number of state-of-the-art competitors on several real-life datasets. Our code is publicly available at {\url{https://github.com/Riwei-HEU/MRL}}.
\end{abstract}




\maketitle

\section{Introduction}

\begin{figure*}[t]
  \centering
  \includegraphics[width=0.85\linewidth]{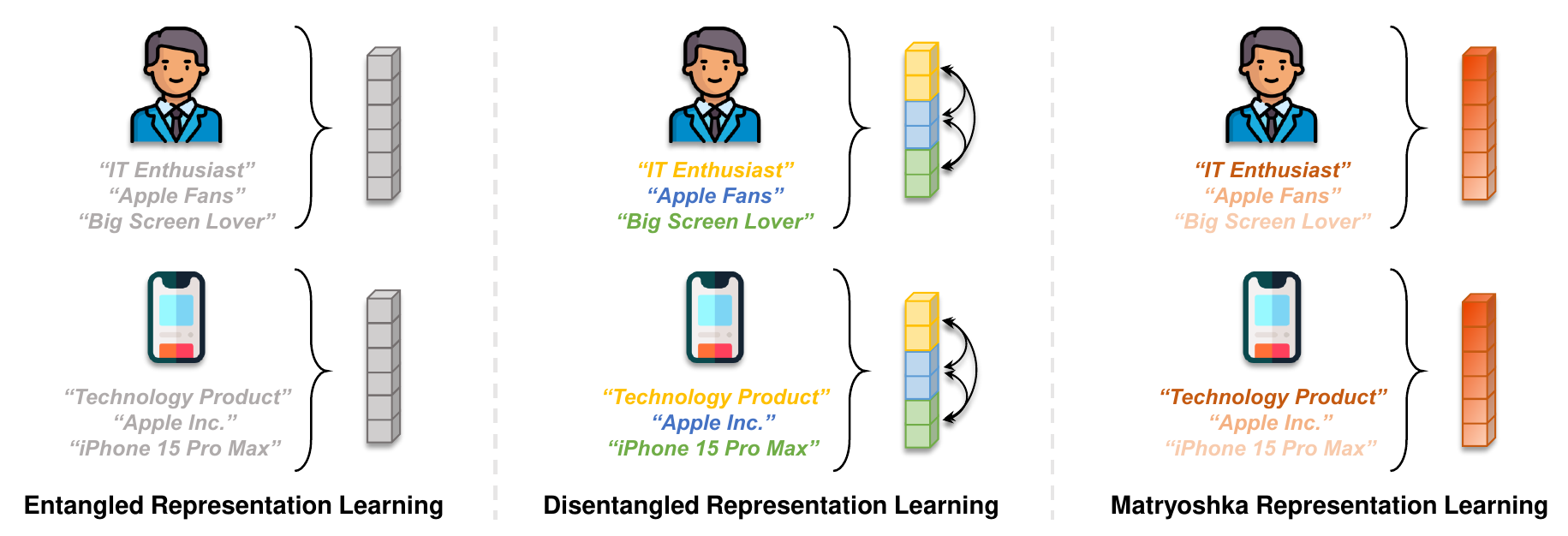}
  \caption{Three representation learning methods for DNN-based RS (best viewed in color). Matryoshka Representation Learning is our proposed method.}
  \label{fig:rl}
  \vspace{-3mm}
\end{figure*}

Recommender systems (RS) have been widely adopted to alleviate information overload in various real-world applications, such as social media~\cite{CW21}, e-commerce~\cite{WDH20}, and online advertising~\cite{ZMF19}. In recent years, empowered by the deep neural networks (DNN), RS has achieved remarkable improvements in recommendation accuracy~\cite{HDK20}, diversity~\cite{LLW19}, and explainability~\cite{LZL20}. In DNN-based RS, learned user and item representations stand as the fundamental components~\cite{WJZ20,HZC22,WYM22}. For example, DNN-based collaborative filtering (CF) expertly maps user preferences and item features into fixed-dimensional user and item vectors (\textit{a.k.a.} embeddings), and subsequently leverages these vectors to predict the ratings of uninteracted items by the users, thereby providing personalized recommendations~\cite{WHW19, HDK20}. 

Undoubtedly, the performance of recommendations is highly dependent on the quality of learned user and item representations. According to the way user preferences and item features are mapped to the embedding vectors, we can divide existing representation learning methods into two primary categories: \textit{entangled} and \textit{disentangled}. As illustrated in the left of Figure~\ref{fig:rl}, \textbf{Entangled Representation Learning} treats each user preference or item feature equally and tends to diffuse it across the entire vector~\cite{RFG09,WWG22,WYM22}. This approach ensures that the learned vector captures a broad range of user preferences or item features. However, it may potentially distribute the information in a way that would dilute the strength of individual preferences or unique features, resulting in the \textit{over-generalization} problem that the primary preference or feature may dominate the entire vector~\cite{LGL20}. For instance, in Figure~\ref{fig:rl}, the learned user vector using entangled representation learning might predominantly reflect the user's potential preference for \textsf{"Apple Fans"}, while inadvertently neglecting the user's inclination towards \textsf{"Big Screen Lover"}. 

To mitigate the over-generalization problem, \textbf{Disentangled Representation Learning} seeks to partition user preferences or item features into a number of distinct, non-overlapping clusters. Each cluster is then mapped to its own niche within a vector space~\cite{WJZ20,LCZ23,ZGL21}. As shown in the middle of Figure~\ref{fig:rl}, the vectors learned by disentangled representation learning are more adept at reflecting the multifaceted user preferences and diverse item features. However, this type of approach comes with a range of substantial \textit{constraints}, preventing it from realizing optimal representations. First, disentangled representation learning requires that there is no overlap between different user preferences or item features, which might not be valid in real situations. For example, in Figure~\ref{fig:rl}, we can see an obvious intersection between \textsf{"IT Enthusiast"} and \textsf{"Apple Fans"}, contrary to the expectation of clear-cut vector spaces. Second, the rigid structure imposed by disentangled representation learning could lead to an \textit{incomplete} (\textit{e.g.,} improperly condensing multiple preferences beyond three into only three segments) or \textit{skewed} (\textit{e.g.,} inappropriately expanding two or fewer preferences to fit three segments) understanding of user preferences or item features, resulting in suboptimal representations.

In light of the limitations inherent in entangled and disentangled representation learning methods, our objective is to develop an enhanced representation learning method that can yield higher-quality vectors of users and items. To accomplish this, it is essential to identify a better structure that can more precisely align with user preferences and item features in the real world: 1). \textit{Users typically favor expressing their preferences in a progressive manner, starting from broad, coarse-grained interests towards precise, fine-grained details}~\cite{SGL24}, as opposed to purely focusing on disparate details from the outset. 2). \textit{Items are usually organized with hierarchical features in real applications}~\cite{CYW19,WCY20}. For instance, the item shown in Figure~\ref{fig:rl} is hierarchically labeled with a broad category of \textsf{"Technology Product"}, the brand \textsf{"Apple Inc."}, and the specific model \textsf{"iPhone 15 Pro Max"}. For these reasons, we tend to formulate user preferences and item features in a hierarchical structure and map each hierarchy into vector spaces accordingly. With this objective, we propose a new representation learning method called \textbf{\underline{M}atryoshka \underline{R}epresentation \underline{L}earning} for \underline{Rec}ommendation (MRL4Rec), as shown in the right of Figure~\ref{fig:rl}. 

MRL4Rec has two major specialties: 1). \textit{It can explicitly represent hierarchical user preferences and item features in vector spaces even without prior knowledge regarding these.} Inspired by the recent work~\cite{KBR22} in machine learning, we propose to restructure user and item vectors into a series of incrementally dimensional and overlapping vector spaces that collectively resemble the nested architecture of a matryoshka, thus names matryoshka representations. The core vector space with the smallest dimension size indicates the broad outlines of user preferences or item features. As the dimension size expands, the vector space progressively becomes tailored to represent more specialized and detailed aspects. 2). \textit{It can ensure that the learned matryoshka representations are able to capture hierarchical user preferences and item features without any explicit supervision.} We broaden the scope of the entangled representation learning task (\textit{e.g.,} Bayesian personalized ranking~\cite{RFG09}) to each vector space in matryoshka representations as the additional supervision of learning user preferences and item features at that hierarchical level. We theoretically demonstrate that constructing specific training triplets tailored to each hierarchical level is necessary to ensure accurate matryoshka representation learning in capturing hierarchical user preferences and item features. Based on the analyses, we further develop a matryoshka negative sampling mechanism to construct training triplets.

We have finally conducted extensive experiments on several real-life datasets. The comparison of MRL4Rec with the state-of-the-art entangled and disentangled representation learning methods shows that it can achieve significant improvements in terms of multiple recommendation metrics.

We summarize our main contributions as follows:
\begin{itemize}[leftmargin=*]
    \item We propose a new representation learning method MRL4Rec, which captures hierarchical user preferences and item features within well-designed matryoshka representations.
    \item We design a novel matryoshka negative sampling mechanism to ensure the matryoshka representation learning process.
    \item Our experiments demonstrate the superior performance of MRL4Rec against the state-of-the-art representation learning methods.
\end{itemize}

\section{Proposed Method}

\subsection{Problem Formulation}
The primary objective of representation learning in recommender systems is to map user preferences and item features into fixed-dimensional user and item vectors. For the ease of exposition, we formulate the representation learning process in the scenario of implicit collaborative filtering. 

Let $\mathcal{U}$ and $\mathcal{I}$ be the set of users and the set of items, respectively. We denote the set of observed user-item interactions by $\mathcal{O} = \{(u, i) \mid u \in \mathcal{U}, i \in \mathcal{I}\}$, where each pair $(u, i)$ represents an interaction between user $u$ and item $i$. Normally, the pair $(u, i)$ indicates user $u$'s preferences towards item $i$, corresponding to a positive training signal. As for the negative training signal, given the pair $(u, i)$, a negative sampling strategy will identify an item $j$ that has not been previously interacted by user $u$, to form the negative pair $(u, j)$. The user and item vectors are then optimized to give positive pairs higher scores than negative pairs by the Bayesian personalized ranking (BPR) loss function~\cite{RFG09}:
\begin{equation}
    \mathcal{L}_{\mathrm{BPR}} = \sum_{(u, i, j)}  - \ln \sigma (\mathbf{e}_u^\top \mathbf{e}_i - \mathbf{e}_u^\top \mathbf{e}_j),
\end{equation}
where $\mathbf{e}_u$, $\mathbf{e}_i$, and $\mathbf{e}_j \in \mathbb{R}^d$ are the $d$-dimensional vectors of user $u$, positive item $i$, and negative item $j$, respectively, the inner product is used to measure the score of positive and negative pairs, and $\sigma (\cdot)$ is the sigmoid function.

\subsection{Matryoshka Representation Learning}
In practical scenarios, user preferences and item features are typically organized in a hierarchical manner. Take the user in Figure~\ref{fig:mrl} as an example, we can clearly observe that the user's preferences are structured according to a hierarchy of \textsf{"IT Enthusiast --> Apple Fans --> Big Screen Lover"}. To effectively capture such hierarchical patterns, it is crucial to explicitly represent them within user and item vectors that can accurately display their intricate relationships. Drawing inspiration from the recent advancements in machine learning, specifically the work of MRL~\cite{KBR22}, we propose to slice user and item vectors into a series of incrementally dimensional and overlapping vector spaces (called matryoshka representations), which enable us to explicitly represent the user preferences and item features at each level individually. For instance, as shown in Figure~\ref{fig:mrl}, the top-2 vector space of the user's vector represents the general aspects of user preferences regarding \textsf{"IT Enthusiast"}. With the increasing dimension size, the top-6 vector space becomes specialized to reflect the user's detailed preferences including \textsf{"Apple Fans"} and \textsf{"Big Screen Lover"}. 

Specifically, following the work~\cite{KBR22}, we consider a set of dimension sizes $\mathcal{D} = \{d_1, d_2, \cdots, d_L\}$, where $d_l$ represents the dimension size of the vector space at the corresponding level $l$. These sizes satisfy the condition of $d_1 < d_2 < \cdots < d_L = d$, ensuring that the dimension sizes gradually increase as we ascend through the different levels. Once we have determined the set of dimension sizes, we can proceed to formulate the sliced vector spaces as:
\begin{equation}
    \mathbf{e}_u^l =  \mathbf{e}_u [0:d_l), \quad \mathbf{e}_i^l =  \mathbf{e}_i [0:d_l),
    \label{eq:sv}
\end{equation}
where $\mathbf{e}_u^l$ and $\mathbf{e}_i^l \in \mathbb{R}^{d_l}$ are the $d_l$-dimensional vectors representing user $u$'s preferences and item $i$'s features at level $l$, respectively. Similarly, we can obtain the vector $\mathbf{e}_j^l$ for item $j$'s features at level $l$.

After acquiring user and item vectors at all levels, we further introduce the matryoshka representation learning (MRL) loss function to optimize the matryoshka representation learning process. Specifically, following the instruction in~\cite{KBR22}, we independently compute the BPR loss for each level by utilizing the corresponding user and item vectors, and subsequently integrate these losses together to optimize the representation learning process. The MRL loss is formulated as:
\begin{equation}
    \mathcal{L}_{\mathrm{MRL}} = \sum_{(u, i, j)} \sum_{l=1}^L - w_l \cdot \ln \sigma (\mathbf{e}_u^{l\top} \mathbf{e}_i^l - \mathbf{e}_u^{l\top} \mathbf{e}_j^l),
\end{equation}
where $w_l \geq 0$ denotes the weight of the loss at level $l$. It can be treated as a hyperparameter to be tuned manually, or as a model parameter to be optimized automatically. In our experiments, we set $w_l$ uniformly as $1/L$ to keep the MRL loss simplicity.

Although MRL loss demonstrates impressive results on various fully supervised representation learning tasks in the original paper~\cite{KBR22}, in the context of recommender systems, it still remains uncertain whether the learned matryoshka representations of users and items could correctly capture hierarchical user preferences and item features. Therefore, we conduct in-depth analyses on $\mathcal{L}_{\mathrm{MRL}}$ and demonstrate that this implementation \textit{cannot} guarantee matryoshka representation learning with hierarchical structures. The key idea is to identify whether there are disparities between $\mathcal{L}_{\mathrm{BPR}}$ and $\mathcal{L}_{\mathrm{MRL}}$ in the direction of vector optimization. Detailed proof processes are as follows:

\begin{figure}[t]
  \centering
  \includegraphics[width=\linewidth]{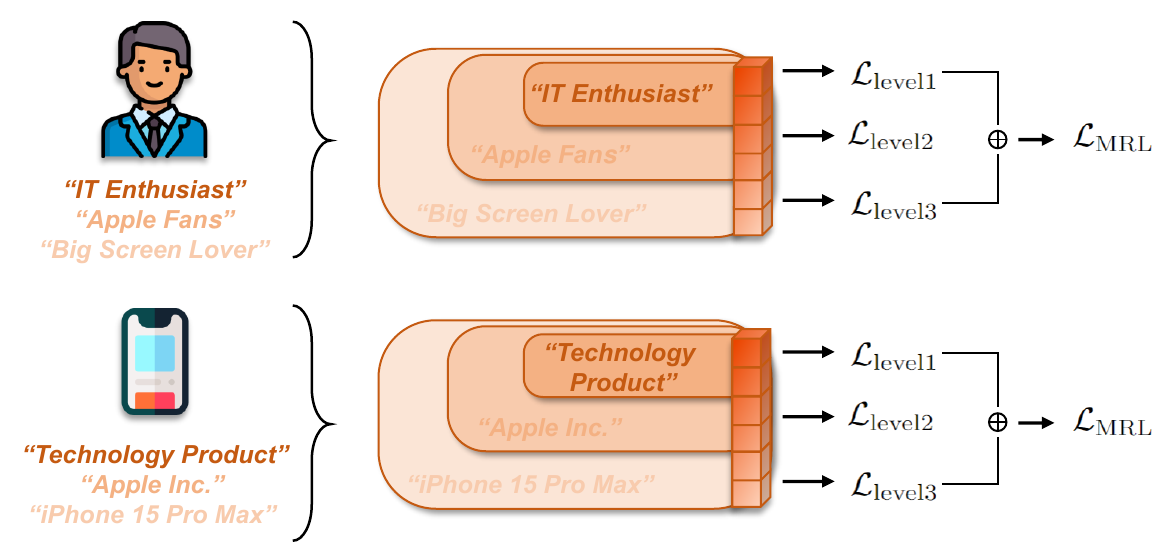}
  \caption{An illustration of matryoshka representation learning process.}
  \label{fig:mrl}
  \vspace{-3mm}
\end{figure}

\begin{theorem}
The representations learned by $\mathcal{L}_{\mathrm{MRL}}$ and $\mathcal{L}_{\mathrm{BPR}}$ exhibit an identical structure.
\label{thm:c1}
\end{theorem}

\begin{proof}
Consider the training triplet $(u, i, j)$, $\mathbf{e}_u$, $\mathbf{e}_i$, and $\mathbf{e}_j \in \mathbb{R}^d$ are the $d$-dimensional vectors of user $u$, positive item $i$, and negative item $j$, respectively. Let $x,y \in \mathbb{N}, x < y \leq L$, we define $\mathbf{e}_*^{x,y}$ as the sliced vector space that spans from dimension $d_x$ to $d_y$, where $*$ could be $u$, $i$, or $j$:
\begin{equation}
    \mathbf{e}_*^{x,y} =  \mathbf{e}_* [d_x:d_y).
    \label{eq:d12}
\end{equation}
In particular, we set $d_0 = 0$. In this case, we simplify $\mathbf{e}_*^{0,y}$ as:
\begin{equation}
    \mathbf{e}_*^{y} =  \mathbf{e}_* [0:d_y).
    \label{eq:d02}
\end{equation}

Given $\mathbf{e}_u^{x,y}$ and $\mathbf{e}_i^{x,y}$, we then define $r_{u,i}^{x,y}$ as the inner product between them:
\begin{equation}
    r_{u,i}^{x,y} =  \mathbf{e}_u^{x,y\top} \mathbf{e}_i^{x,y}.
\end{equation}

For the inner product between $\mathbf{e}_u^{y}$ and $\mathbf{e}_i^{y}$, similarly we have:
\begin{equation}
    r_{u,i}^{y} =  \mathbf{e}_u^{y\top} \mathbf{e}_i^{y}.
    \label{eq:rui}
\end{equation}

According to Equation~(\ref{eq:d12})-(\ref{eq:rui}) and the definition of inner product, we can rewrite $\mathcal{L}_{\mathrm{BPR}}$ as:
\begin{align}
    \mathcal{L}_{\mathrm{BPR}} &= \sum_{(u, i, j)} - \ln \sigma (\mathbf{e}_u^\top \mathbf{e}_i - \mathbf{e}_u^\top \mathbf{e}_j) = \sum_{(u, i, j)}  - \ln \sigma (r_{u,i}^{L} - r_{u,j}^{L}) \notag \\
    &= \sum_{(u, i, j)} - \ln \sigma ((r_{u,i}^{1} + \cdots + r_{u,i}^{L-1,L}) - (r_{u,j}^{1} + \cdots + r_{u,j}^{L-1,L})) \notag \\
    &= \sum_{(u, i, j)} - \ln \sigma ((r_{u,i}^{1} - r_{u,j}^{1}) + \cdots + (r_{u,i}^{L-1,L} - r_{u,j}^{L-1,L})) \notag \\
    &= \sum_{(u, i, j)} - \frac{1}{L} \cdot L \cdot \ln \sigma ((r_{u,i}^{1} - r_{u,j}^{1}) + \cdots + (r_{u,i}^{L-1,L} - r_{u,j}^{L-1,L})).
\end{align}

Consider a specific vector space, for example, $\mathbf{e}_u^{l-1,l}$, we can obtain the optimization effect of $\mathcal{L}_{\mathrm{BPR}}$ on this vector by taking its partial derivative with respect to $r_{u,i}^{l-1,l}$:
\begin{align}    
    \frac{\partial \mathcal{L}_{\mathrm{BPR}}}{\partial r_{u,i}^{l-1,l}} = \frac{\partial (- \frac{1}{L} \cdot L \cdot \ln \sigma ((r_{u,i}^{1} - r_{u,j}^{1}) + \cdots + (r_{u,i}^{L-1,L} - r_{u,j}^{L-1,L})))}{\partial r_{u,i}^{l-1,l}}.
    \label{eq:pd}
\end{align}

For simplicity, we freeze the remaining vector spaces by keeping them as initialized values throughout the learning process, therefore, we can update Equation~(\ref{eq:pd}) as:
\begin{align}
    \frac{\partial \mathcal{L}_{\mathrm{BPR}}}{\partial r_{u,i}^{l-1,l}} &\approx \frac{\partial (- \frac{1}{L} \cdot L \cdot \ln \sigma (0 + \cdots (r_{u,i}^{l-1,l} - r_{u,j}^{l-1,l}) + \cdots))}{\partial r_{u,i}^{l-1,l}} \notag \\
    & = \frac{\partial (- \frac{1}{L} \cdot L \cdot \ln \sigma (r_{u,i}^{l-1,l} - r_{u,j}^{l-1,l}))}{\partial r_{u,i}^{l-1,l}} \notag \\
    & = - \frac{1}{L} \cdot L \cdot (1 - \sigma (r_{u,i}^{l-1,l} - r_{u,j}^{l-1,l})).
    \label{eq:bprpd}
\end{align}

Similarly, according to Equation~(\ref{eq:d12})-(\ref{eq:rui}) and the definition of inner product, we can rewrite $\mathcal{L}_{\mathrm{MRL}}$ as:
\begin{align}
    \mathcal{L}_{\mathrm{MRL}} &= \sum_{(u, i, j)} \sum_{l=1}^L - \frac{1}{L} \cdot \ln \sigma (\mathbf{e}_u^{l\top} \mathbf{e}_i^l - \mathbf{e}_u^{l\top} \mathbf{e}_j^l) \notag \\
    &=  \sum_{(u, i, j)} \sum_{l=1}^L - \frac{1}{L} \cdot \ln \sigma (r_{u,i}^{l} - r_{u,j}^{L}) \notag \\
    &= \sum_{(u, i, j)} - \frac{1}{L} \cdot (\ln \sigma (r_{u,i}^{1} - r_{u,j}^{1}) + \cdots + \ln \sigma(r_{u,i}^{L} - r_{u,i}^{L})) \notag \\
    &= \sum_{(u, i, j)} - \frac{1}{L} \cdot (\ln \sigma (r_{u,i}^{1} - r_{u,j}^{1}) + \cdots + \ln \sigma ((r_{u,i}^{1} - r_{u,j}^{1}) \notag \\
    &\qquad \qquad \qquad + \cdots + (r_{u,i}^{L-1,L} - r_{u,j}^{L-1,L}))).
    \label{eq:rwmrl}
\end{align}

To get the optimization effect of $\mathcal{L}_{\mathrm{MRL}}$ on the same vector $\mathbf{e}_u^{l-1,l}$, we freeze the remaining vector spaces and calculate the partial derivative of $\mathcal{L}_{\mathrm{MRL}}$ with respect to $r_{u,i}^{l-1,l}$:
\begin{align}
    &\frac{\partial \mathcal{L}_{\mathrm{MRL}}}{\partial r_{u,i}^{l-1,l}} \approx \frac{\partial (- \frac{1}{L} \cdot (\ln \sigma (0) + \cdots \ln \sigma (0 + \cdots (r_{u,i}^{l-1,l} - r_{u,j}^{l-1,l}) + \cdots))}{\partial r_{u,i}^{l-1,l}} \notag \\
    & = - \frac{1}{L} \cdot (0 + \cdots \underbrace{(1 - \sigma (r_{u,i}^{l-1,l} - r_{u,j}^{l-1,l})) + \cdots (1 - \sigma (r_{u,i}^{l-1,l} - r_{u,j}^{l-1,l}))}_{L-l+1})\notag \\
    & = - \frac{1}{L} \cdot (L-l+1) \cdot (1 - \sigma (r_{u,i}^{l-1,l} - r_{u,j}^{l-1,l})).
    \label{eq:mrlpd}
\end{align}

Comparing Equation~(\ref{eq:bprpd}) and Equation~(\ref{eq:mrlpd}), we can observe that given the same training triplet $(u, i, j)$, the optimization effect of $\mathcal{L}_{\mathrm{BPR}}$ and $\mathcal{L}_{\mathrm{MRL}}$ on the identical vector space differs only in terms of magnitude, which is corresponding to the coefficient (\textit{i.e.,} $L$ and $L-l+1$); while yields no disparity in the direction, which is corresponding to the triplet $(u, i, j)$ (\textit{i.e.,} $r_{u,i}^{l-1,l} - r_{u,j}^{l-1,l}$). Therefore, the representations learned by $\mathcal{L}_{\mathrm{MRL}}$ and $\mathcal{L}_{\mathrm{BPR}}$ exhibit an identical structure -- the above completes the proof.
\end{proof}

\subsection{Matryoshka Negative Sampling}

\begin{algorithm}[t]
    \caption{Matryoshka Negative Sampling (MNS)}
    \KwIn{Set of observed interactions $\mathcal{O} = \{(u, i)\}$, set of dimension sizes $\mathcal{D} = \{d_1, d_2, \cdots, d_L\}$, hard negative sampling strategy $f$.}
    \KwOut{Set of training tuples $\mathcal{T} = \{(u, i, j_1, \cdots, j_L)\}$.}
     Initialize an empty set $\mathcal{T} = \{\}$ for training tuples. \\

    \For{$each ~ (u, i) \in \mathcal{O}$}
    {
        Get the vectors $\mathbf{e}_u$, $\mathbf{e}_i$ of $u$, $i$. \\
        \For{$each ~ d_l \in \mathcal{D}$}
        {   Get the sliced vectors $\mathbf{e}_u^l$, $\mathbf{e}_i^l$ of $u$, $i$ by Equation~(\ref{eq:sv}). \\
            Sample an item $j_l$ by $f$ considering $\mathbf{e}_u^l$, $\mathbf{e}_i^l$. \\
            \eIf{$l = 1$}
            {
                Add $j_l$ to $(u, i)$.
            }{
            Add $j_l$ to $(u, i, \cdots, j_{l-1})$.
            }
        }
    Add $(u, i, j_1, \cdots, j_L)$ to $\mathcal{T}$. \\
    }
\label{al:mns}
\end{algorithm}

Based on the theoretical analyses conducted, in order to achieve matryoshka representation learning with hierarchical structures, it is essential to construct specific training triplets tailored to each hierarchical level. Therefore, we further propose the matryoshka negative sampling (MNS) mechanism for matryoshka representation learning, whose entire procedure is detailed in Algorithm~\ref{al:mns}. 

Specifically, we identify an existing hard negative sampling strategy $f$ as the foundational negative sampling method. And for each observed user-item interaction $(u,i)$, we utilize $f$ to select the negative item $j_l$ corresponding to level $l$ by taking as input the sliced user and item vectors $\mathbf{e}_u^l$ and $\mathbf{e}_i^l$. After sampling negative items for all levels, we obtain the final tuple $(u, i, j_1, \cdots, j_L)$ for training. Consequently, the MRL loss with MNS (MRL4Rec) can be formulated as follows:
\begin{equation}
    \mathcal{L}_{\mathrm{MRL_{w/MNS}}} = \sum_{(u, i, j_1, \cdots, j_L)} \sum_{l=1}^L - w_l \cdot \ln \sigma (\mathbf{e}_u^{l\top} \mathbf{e}_i^l - \mathbf{e}_u^{l\top} \mathbf{e}_{j_l}^l).
\end{equation}
Considering both the effectiveness and efficiency, we choose the simplest hard negative sampling strategy, dynamic negative sampling (DNS)~\cite{ZCW13}, as $f$ in our experiments.

Moreover, we can get the optimization effect of $\mathcal{L}_{\mathrm{MRL_{w/MNS}}}$ on the vector $\mathbf{e}_u^{l-1,l}$ by conducting similar calculations as Equation~(\ref{eq:rwmrl}) and (\ref{eq:mrlpd}). Finally, we have the partial derivative of $\mathcal{L}_{\mathrm{MRL_{w/MNS}}}$ with respect to $r_{u,i}^{l-1,l}$, \textit{i.e.,} $\frac{\partial \mathcal{L}_{\mathrm{MRL_{w/MNS}}}}{\partial r_{u,i}^{l-1,l}}$, as follows:

$- \frac{1}{L} \cdot (0 + \cdots \underbrace{(1 - \sigma (r_{u,i}^{l-1,l} - r_{u,j_l}^{l-1,l})) + \cdots (1 - \sigma (r_{u,i}^{l-1,l} - r_{u,j_L}^{l-1,l}))}_{L-l+1})$
It is clear that benefit from the specific training triplet tailored to each individual level, there are significant disparities between $\mathcal{L}_{\mathrm{BPR}}$ and $\mathcal{L}_{\mathrm{MRL_{w/MNS}}}$ in the direction of vector optimization. In addition, vectors at the higher level (\textit{i.e.,} smaller $l$) will receive more optimization effects, motivating them to learn more generalized representations, which is in line with our expectations.

\section{Experiments}
In this section, we perform extensive experiments to evaluate our proposed MRL4Rec and MNS. Specifically, we try to answer the following research questions:
\begin{itemize}[leftmargin=*]
    \item \textbf{RQ1:} How does MRL4Rec perform compared with previous entangled and disentangled representation learning methods?
    \item \textbf{RQ2:} Does MNS benefit the matryoshka representation learning with hierarchical structures experimentally?
    \item \textbf{RQ3:} Can representations learned by MRL4Rec capture hierarchical user preferences and item features?
    \item \textbf{RQ4:} What are the impacts of the number of hierarchical levels?
\end{itemize}

\subsection{Experimental Setup}
\subsubsection{Datasets and Evaluation Metrics.}
We consider three widely used real-life benchmark datasets in experiments: \textbf{Amazon-Phones}, \textbf{Amazon-Sports}, and \textbf{Amazon-Tools}. These datasets are adopted from the Amazon review dataset\footnote{https://jmcauley.ucsd.edu/data/amazon/} with different categories, \textit{e.g.,} phones and sports, containing ratings, reviews, and helpfulness votes of real users from online shopping platform Amazon~\cite{HM16}. We only use the review data and treat each review as an implicit user-item interaction. Following~\cite{HDK20, SCF23}, we randomly split each user’s interactions into training/test sets with a ratio of 80\%/20\%, and build the validation set by randomly sampling 10\% interactions of the training set. Table~\ref{tab:data} summarizes the statistics of the three datasets. We report the recommendation performances in terms of $Recall@K$ and $NDCG@K$, with $K = \{10, 20, 50\}$, where higher values indicate better performances~\cite{WHW19}.

\subsubsection{Baseline Methods.}
We compare our methods with a wide range of representative representation learning (RL) methods:

\noindent \textbf{(i) Entangled RL Methods:}
\begin{itemize}[leftmargin=*]
    \item \textbf{BCE}~\cite{HKV08}: Binary cross-entropy (BCE) loss is a widely used loss function for binary classification tasks. In recommender systems, it aims to refine user and item representations by optimizing their capacity to predict binary user-item interactions effectively.
    \item \textbf{BPR}~\cite{RFG09}: Bayesian personalized ranking (BPR) loss is a typical ranking-based loss function in recommender systems. It aims to refine user and item representations by maximizing the difference between the predicted ratings for positive and negative items.
    \item \textbf{BPR-D}~\cite{ZCW13}: BPR-D constructs training triplets by utilizing the dynamic negative sampling (DNS) strategy, and optimizes the representation learning process through BPR loss.
    \item \textbf{SSM}~\cite{WWG22}: Sampled softmax (SSM) loss is an efficient alternative to traditional softmax loss. It leverages a subset of sampled negative items instead of considering all items in training to preserve the desired ranking property while reducing the training cost.
    \item \textbf{UIB}~\cite{ZZY22}: UIB is a hybrid loss function combining the BCE loss and BPR loss through learning explicit user interest boundaries.
    \item \textbf{DAU}~\cite{WYM22}: DirectAU (DAU) directly optimizes the alignment and uniformity properties of user and item representations, ensuring that these representations are both coherent and evenly distributed in the embedding space.
\end{itemize}

\noindent \textbf{(ii) Disentangled RL Methods:}
\begin{itemize}[leftmargin=*]
    \item \textbf{MVAE}~\cite{MZC19}: MacridVAE (MVAE) aims to learn disentangled representations based on user behavior data. It explicitly models the separation of macro and micro factors, and performs disentanglement at each level.
    \item \textbf{DGCF}~\cite{WJZ20}: DGCF considers user-item relationships at the finer granularity of user intents. It generates an intent-aware graph by modeling a distribution over intents for each interaction and effectively disentangles these intents in the representations of users and items.
    \item \textbf{DENS}~\cite{LCZ23}: DENS disentangles relevant and irrelevant factors of items with contrastive learning and considers these factors to identify the optimal negative items to learn more informative representations of users and items.
    \item \textbf{DCCF}~\cite{RXZ23}: DCCF utilizes disentangled contrastive learning to explore latent factors underlying implicit intents for user-item interactions. It introduces a graph structure learning layer that enables adaptive interaction augmentation based on learned disentangled user and item intent-aware dependencies.
\end{itemize}

We also introduce several variants of our methods utilized in the experiments:

\noindent \textbf{(iii) Matryoshka RL Methods:}
\begin{itemize}[leftmargin=*]
    \item \textbf{BPR-M}: BPR-M incorporates the proposed matryoshka negative sampling (MNS) mechanism to construct training tuples and then optimizes the representation learning process by BPR loss.
    \item \textbf{MRL-D}: MRL-D adopts DNS to construct training triplets and optimizes the representation learning process by MRL loss. 
    \item \textbf{MRL4Rec}: MRL4Rec, also denoted as \textbf{MRL-M}, utilizes MNS to construct training tuples and optimizes the representation learning process by MRL loss.
\end{itemize}

\begin{table}[t]
\centering
\caption{Statistics of the datasets used in the experiments.}
\begin{tabular}{@{}lcccc@{}}
\toprule
\textbf{Dataset} & \textbf{\#Users} & \textbf{\#Items} & \textbf{\#Inters} & \textbf{Density} \\ \midrule
Amazon-Phones & 27,879 & 10,429 & 194,439 & 0.067\% \\ 
Amazon-Sports & 35,598 & 18,357 & 296,337 & 0.045\% \\ 
Amazon-Tools & 16,638 & 10,217 & 134,476 & 0.079\%\\ \bottomrule
\end{tabular}
\label{tab:data}
\end{table}

\subsubsection{Implementation Details.}
We follow the experimental setting in DirectAU~\cite{WYM22} and utilize matrix factorization (MF) as our base recommendation model. For methods utilizing graph neural networks, we set the number of graph layers to 1. The dimension size of user and item vectors is fixed to 64, and all the vectors are initialized with the Xavier method. We optimize the representation learning process with Adam~\cite{KB15} and use the default learning rate of 0.001 and default mini-batch size of 2,048. The number of training epochs is set to 300 and we use the early stopping mechanism to prevent overfitting. We set the set of dimension sizes $\mathcal{D}$ as $\{4, 8, 16, 32, 64\}$ and $w_l$ as $1/5$ for MRL loss. For MNS, we adopt DNS as the basic negative sampling strategy $f$. For all baseline methods, we employ the recommended settings of hyperparameters to maintain consistency and ensure fair comparison.

\begin{table*}[t]
\centering
\caption{Performances of MRL-M (MRL4Rec) and main baseline methods. The best results are in bold, and the second best are underlined. Improvements are calculated over the best baseline method and statistically significant with $p\text{-value} < 0.01$.}
\begin{tabular}{@{}cc*{10}{>{\centering\arraybackslash}p{1cm}}@{}}
\toprule
\multirow{2}{*}{\textbf{Dataset}} & \multirow{2}{*}{\textbf{Metric}} & \multicolumn{5}{c}{\textbf{Entangled RL Method}} & \multicolumn{4}{c}{\textbf{Disentangled RL Method}} & \textbf{Ours} \\ \cmidrule(l){3-7} \cmidrule(l){8-11} \cmidrule(l){12-12} & & BCE & BPR & SSM & UIB & DAU & MVAE & DGCF & DENS & DCCF & MRL-M  \\ \midrule
\multirow{6}{*}{Amazon-Phones} 
& $Recall@10$ & 0.0730 & 0.0760 & 0.0787 & 0.0708  & 0.0791 & 0.0888 & 0.0813 & 0.0851 & \textbf{0.0946} & {\ul 0.0906} \\
& $Recall@20$ & 0.1082 & 0.1106 & 0.1137 & 0.1042  & 0.1156 & 0.1237 & 0.1174 & 0.1216 & \textbf{0.1366} & {\ul 0.1286} \\
& $Recall@50$ & 0.1668 & 0.1695 & 0.1731 & 0.1625  & 0.1802 & 0.1844 & 0.1803 & 0.1849 & \textbf{0.2031} & {\ul 0.1920} \\ \cmidrule(l){2-12}
& $NDCG@10$   & 0.0471 & 0.0497 & 0.0512 & 0.0440  & 0.0485 & 0.0589 & 0.0530 & 0.0562 & \textbf{0.0619} & {\ul 0.0601} \\ 
& $NDCG@20$   & 0.0578 & 0.0598 & 0.0613 & 0.0538  & 0.0593 & 0.0692 & 0.0635 & 0.0668 & \textbf{0.0742} & {\ul 0.0710} \\ 
& $NDCG@50$   & 0.0715 & 0.0746 & 0.0748 & 0.0673  & 0.0742 & 0.0832 & 0.0781 & 0.0813 & \textbf{0.0897} & {\ul 0.0858} \\ \midrule
\multirow{6}{*}{Amazon-Sports}
& $Recall@10$ & 0.0440 & 0.0447 & 0.0483 & 0.0392  & 0.0492 & 0.0454 & 0.0534 & 0.0530 & {\ul 0.0538} & \textbf{0.0582} \\
& $Recall@20$ & 0.0652 & 0.0673 & 0.0708 & 0.0605  & 0.0723 & 0.0639 & 0.0786 & 0.0790 & {\ul 0.0792} & \textbf{0.0845} \\
& $Recall@50$ & 0.1085 & 0.1125 & 0.1164 & 0.1022  & 0.1167 & 0.0996 & 0.1280 & 0.1259 & {\ul 0.1288} & \textbf{0.1318} \\ \cmidrule(l){2-12}
& $NDCG@10$   & 0.0283 & 0.0291 & 0.0313 & 0.0248  & 0.0321 & 0.0322 & 0.0350 & 0.0353 & {\ul 0.0354} & \textbf{0.0388} \\ 
& $NDCG@20$   & 0.0343 & 0.0360 & 0.0380 & 0.0314  & 0.0391 & 0.0378 & 0.0426 & {\ul 0.0435} & 0.0431 & \textbf{0.0469} \\
& $NDCG@50$   & 0.0441 & 0.0468 & 0.0487 & 0.0413  & 0.0498 & 0.0463 & 0.0544 & {\ul 0.0550} & {\ul 0.0550} & \textbf{0.0581} \\ \midrule
\multirow{6}{*}{Amazon-Tools}
& $Recall@10$ & 0.0381 & 0.0383 & 0.0398 & 0.0353  & 0.0451 & 0.0386 & {\ul 0.0501} & 0.0485 & 0.0464 & \textbf{0.0509} \\
& $Recall@20$ & 0.0559 & 0.0553 & 0.0572 & 0.0528  & 0.0673 & 0.0554 & {\ul 0.0717} & 0.0684 & 0.0677 & \textbf{0.0734} \\
& $Recall@50$ & 0.0883 & 0.0869 & 0.0908 & 0.0843  & 0.1035 & 0.0838 & \textbf{0.1122} & 0.1022 & 0.1036 & {\ul 0.1115} \\ \cmidrule(l){2-12}
& $NDCG@10$   & 0.0243 & 0.0247 & 0.0259 & 0.0225  & 0.0295 & 0.0267 & {\ul 0.0329} & 0.0312 & 0.0301 & \textbf{0.0339} \\ 
& $NDCG@20$   & 0.0296 & 0.0299 & 0.0310 & 0.0277  & 0.0360 & 0.0318 & {\ul 0.0394} & 0.0373 & 0.0365 & \textbf{0.0404} \\
& $NDCG@50$   & 0.0376 & 0.0375 & 0.0393 & 0.0354  & 0.0446 & 0.0385 & {\ul 0.0489} & 0.0455 & 0.0450 & \textbf{0.0491} \\ \bottomrule
\end{tabular}
\label{tab:per}
\end{table*}

\subsection{Performance Comparison (RQ1)}
We present the performance results of our proposed MRL-M and main baseline methods in Table~\ref{tab:per}, where the best results are boldfaced and the second-best results are underlined. Our observations are as follows:

Overall, compared with all baseline methods, MRL-M yields considerable performance on three datasets and across all recommendation metrics. Specifically, on the Amazon-Sports dataset, its relative improvements over the strongest baselines are 8.18\%, 6.69\% and 2.34\% in terms of $Recall@10$, $Recall@20$ and $Recall@50$, respectively, and 9.60\%, 8.82\% and 5.64\% in terms of $NDCG@10$, $NDCG@20$ and $NDCG@50$, respectively. This demonstrates that MRL-M is capable of learning higher-quality user and item representations and thus achieves better recommendation results. We attribute such improvements to the capacity of MRL-M to effectively capture hierarchical user preferences and item features.

Among the baseline methods, it is clearly observable that disentangled representation learning methods demonstrate competitive performance, whereas entangled representation learning methods perform relatively poorly. This observation is indeed consonant with the intuition that, in the real world, users typically exhibit a diverse range of preferences, and disentangled representation learning methods stand out in capturing these multifaceted preferences, thus yielding superior recommendation results compared with entangled representation learning methods.

In comparison to disentangled representation learning methods, such as DGCF and DCCF, MRL-M does \textit{not} exhibit significant advantages in datasets like Amazon-Phones and Amazon-Tools. We speculate that this may be attributed to the fact that these two methods are built upon more advanced graph neural networks, which contribute to their superior performance. To more accurately evaluate the effectiveness of MRL-M in representation learning, we migrate MRL-M from its original MF model to the 1-layer LightGCN~\cite{HDK20}. Our experiments on the Amazon-Phone dataset reveal that MRL-M achieves 0.1051, 0.1502, 0.2262 in terms of $Recall@10$, $ Recall@20$, and $Recall@50$, respectively, and 0.0695, 0.0827, 0.1004 in terms of $NDCG@10$, $NDCG@20$, and $NDCG@50$, respectively. When compared to DCCF, it becomes evident that MRL-M performs consistently better. This finding verifies the effectiveness of MRL-M and justifies our initial motivation that learning representations with hierarchical structures aligns more closely with user preferences and item features.

\begin{table*}[t]
\centering
\caption{Performances of BPR-D, BPR-M, MRL-D and MRL-M.}
\begin{tabular}{@{}cccccccccccccc@{}}
\toprule
\multirow{2}{*}{\textbf{Metric}} & \multicolumn{4}{c}{\textbf{Amazon-Phones}} & \multicolumn{4}{c}{\textbf{Amazon-Sports}} & \multicolumn{4}{c}{\textbf{Amazon-Tools}} \\ \cmidrule(l){2-5} \cmidrule(l){6-9} \cmidrule(l){10-13} & BPR-D & BPR-M & MRL-D & MRL-M & BPR-D & BPR-M & MRL-D & MRL-M & BPR-D & BPR-M & MRL-D & MRL-M \\ \midrule
$Recall@10$ & 0.0840 & 0.0775 & 0.0841 & 0.0906 & 0.0527 & 0.0499 & 0.0523 & 0.0582 & 0.0470 & 0.0409 & 0.0479 & 0.0509 \\
$Recall@20$ & 0.1208 & 0.1099 & 0.1211 & 0.1286 & 0.0774 & 0.0721 & 0.0766 & 0.0845 & 0.0666 & 0.0592 & 0.0680 & 0.0734 \\
$Recall@50$ & 0.1793 & 0.1647 & 0.1833 & 0.1920 & 0.1220 & 0.1134 & 0.1230 & 0.1318 & 0.0977 & 0.0901 & 0.1017 & 0.1115 \\  \midrule
$NDCG@10$   & 0.0556 & 0.0510 & 0.0552 & 0.0601 & 0.0350 & 0.0334 & 0.0351 & 0.0388 & 0.0312 & 0.0267 & 0.0316 & 0.0339 \\ 
$NDCG@20$   & 0.0664 & 0.0602 & 0.0660 & 0.0710 & 0.0425 & 0.0401 & 0.0423 & 0.0469 & 0.0378 & 0.0322 & 0.0377 & 0.0404 \\ 
$NDCG@50$   & 0.0799 & 0.0730 & 0.0804 & 0.0858 & 0.0532 & 0.0500 & 0.0537 & 0.0581 & 0.0452 & 0.0395 & 0.0457 & 0.0491 \\ \bottomrule
\end{tabular}
\label{tab:ab}
\end{table*}

\begin{figure*}[t]
  \centering
  \hspace{\fill}
  \subfigure{\includegraphics[width=0.3\linewidth]{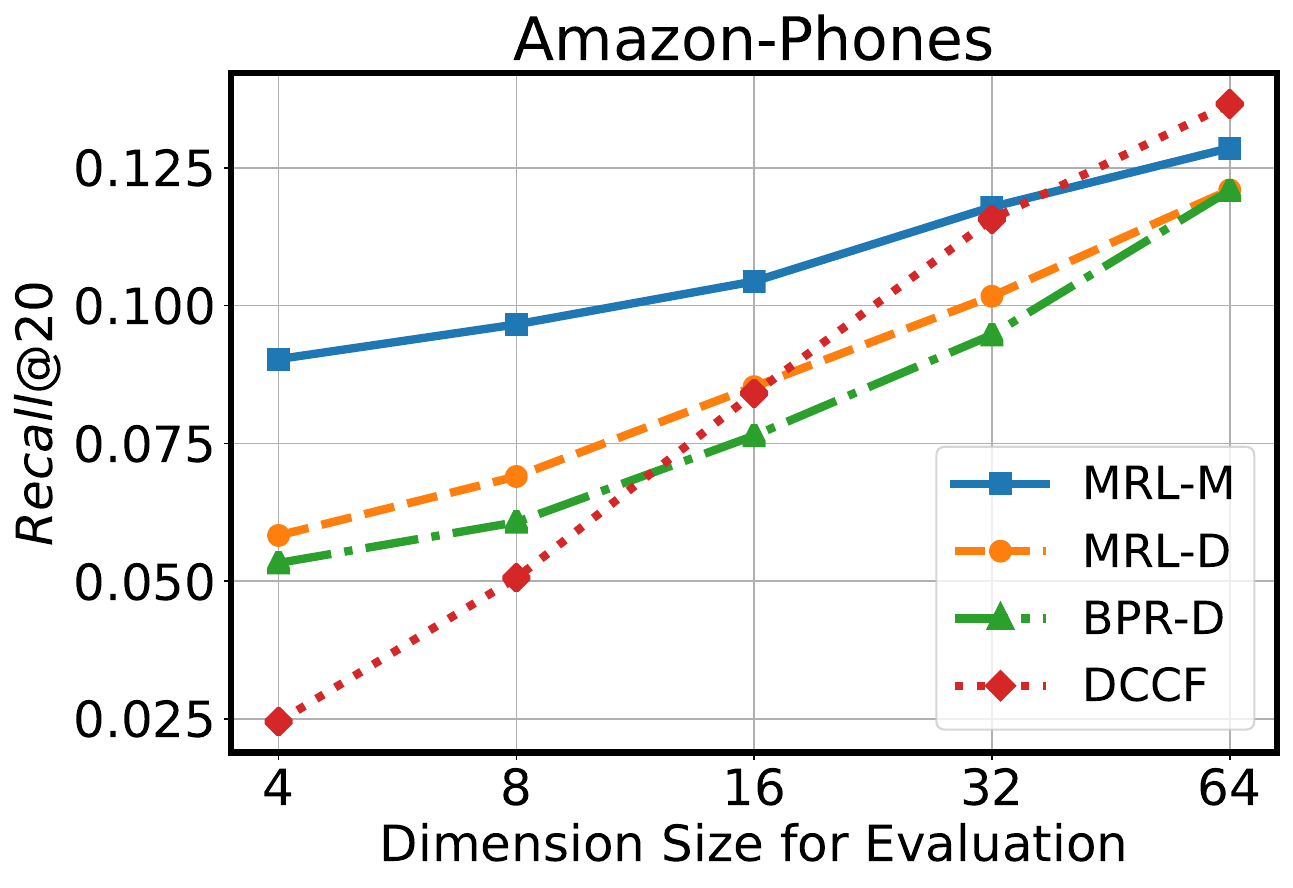} \label{fig:sp}}
  \hspace{\fill}
  \subfigure{\includegraphics[width=0.3\linewidth]{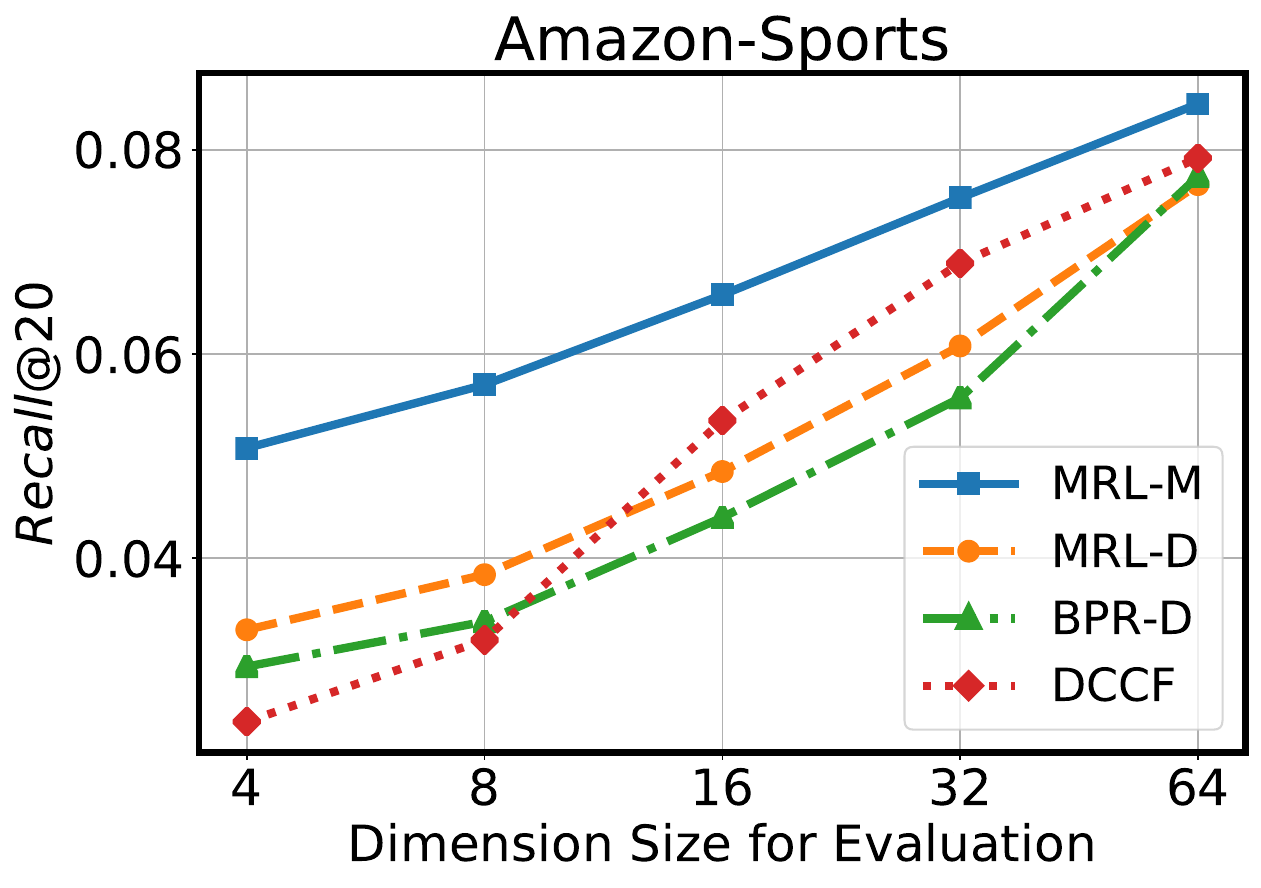} \label{fig:ss}}
  \hspace{\fill}
  \subfigure{\includegraphics[width=0.3\linewidth]{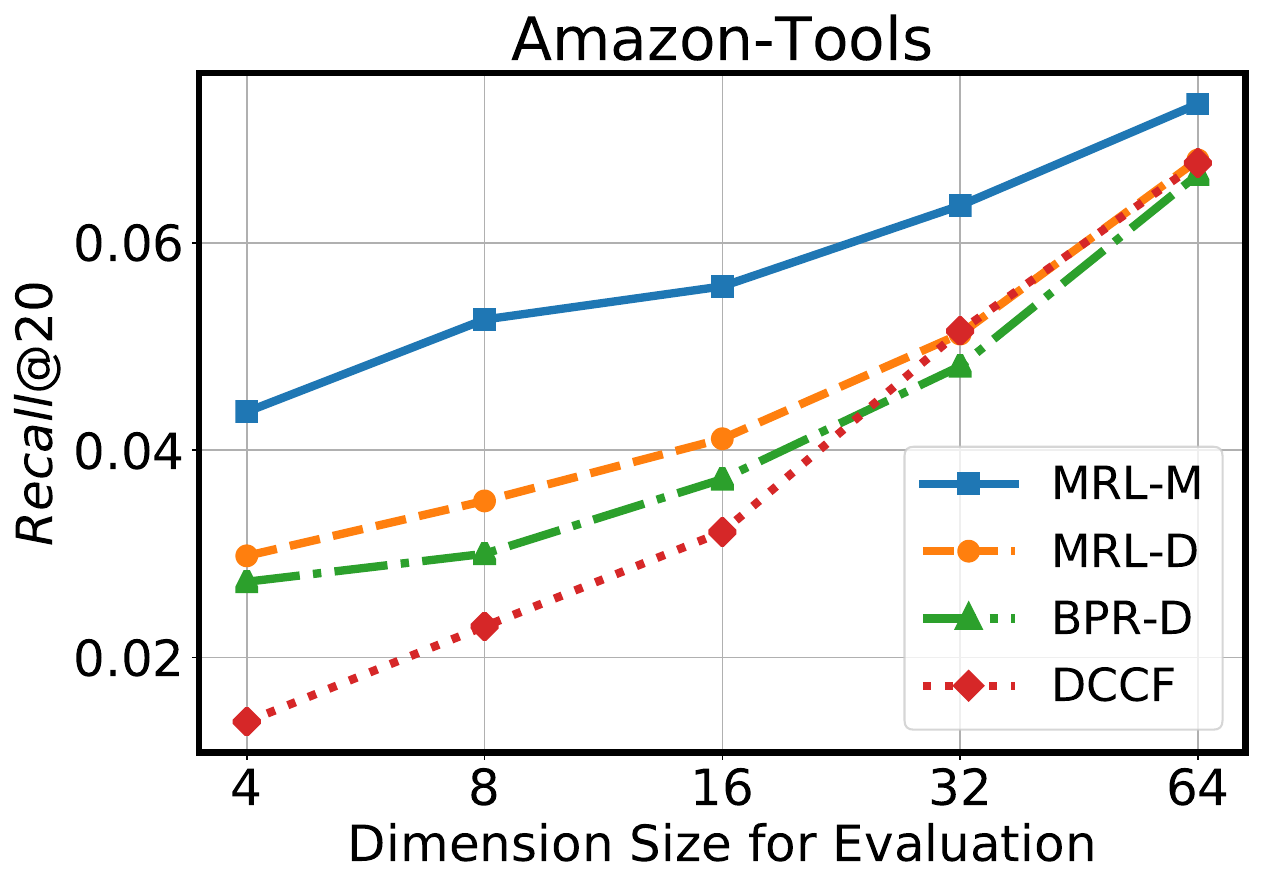} \label{fig:st}}
  \hspace{\fill}
  \caption{$Recall@20$ of MRL-M, MRL-D, BPR-D, and DCCF under different dimension sizes for evaluation.}
  \label{fig:ea}
\end{figure*}

\subsection{Ablation Study (RQ2)}
To validate the consistency of the proposed matryoshka negative sampling (MNS) mechanism with our theoretical analyses and its effectiveness in ensuring matryoshka representation learning with hierarchical structures. We design several variants of our methods, \textit{i.e.,} BPR-M, MRL-D, and MRL-M, and report the performance results of these variants and BPR-D in Table~\ref{tab:ab}. We can draw a few interesting observations as follows:

\noindent \textbf{(i) Comparing BPR-D and BPR-M:}
\begin{itemize}[leftmargin=*]
    \item When directly comparing the performance of BPR-D and BPR-M, we clearly see that BPR-M performs worse than BPR-D on all datasets and across all metrics. This suggests that, contrary to what might be expected, the introduction of the MNS mechanism does not improve performance when in conjunction with the Bayesian personalized ranking (BPR) loss.
    \item A possible explanation for this observation could be that the BPR loss is specifically tailored for learning entangled representations, whereas the MNS mechanism, which was originally designed for learning matryoshka representations, may introduce unnecessary complexities or noise that are not conducive to the entangled representation learning process.
    \item It is worth noting that this observation underscores the importance of carefully matching sampling strategies and loss functions to ensure learning optimal representations.
\end{itemize}

\noindent \textbf{(ii) Comparing BPR-D and MRL-D:}
\begin{itemize}[leftmargin=*]
    \item When comparing the results of BPR-D and MRL-D, our observation reveals that there is little to no discernible difference in performance between BPR-D and MRL-D across all datasets and metrics.
    \item This finding aligns with our theoretical understanding that merely swapping out the loss function to the matryoshka representation learning (MRL) loss without modifying the training triples is insufficient for learning matryoshka representations that exhibit hierarchical structures.
\end{itemize}

\noindent \textbf{(iii) Comparing MRL-D and MRL-M:}
\begin{itemize}[leftmargin=*]
    \item The comparison between MRL-D and MRL-M clearly demonstrates the superiority of MRL-M across all datasets and metrics.
    \item This observation provides strong empirical evidence in support of our hypothesis that integrating the MNS mechanism with the MRL loss significantly boosts the effectiveness of learning matryoshka representations with hierarchical structures.
\end{itemize}

\subsection{Effectiveness Analysis (RQ3)}
Given the absence of ground truth data regarding hierarchical user preferences and item features, conclusively verifying whether MRL-M effectively maps these aspects into learned matryoshka representations accordingly remains challenging. Therefore, as a compromise, we attempt to verify whether each sliced vector in matryoshka representations, which is expected to represent a particular hierarchical level of user preferences or item features, contains meaningful information beneficial for recommendation. Figure~\ref{fig:ea} illustrates the $Recall@20$ results of MRL-M, MRL-D, BPR-D, and DCCF when evaluated using only top-$n$ dimension vectors. Here $n$ takes values from the set of dimension sizes $\mathcal{D}$. We observe that:
\begin{itemize}[leftmargin=*]
    \item MRL-M demonstrates superior performance in most cases, especially under smaller dimension sizes for evaluation. Surprisingly, even when evaluated using only top-8 dimension vectors, MRL-M can attain comparable performance to other methods utilizing top-32 dimension vectors. This observation highlights the remarkable effectiveness of MRL-M in capturing relevant information, even with a significantly reduced vector space.
    \item MRL-D and BPR-D perform similarly under different dimension sizes for evaluation. This observation further proves the correctness of Theorem~\ref{thm:c1}, \textit{i.e.,} the representations learned by $\mathcal{L}_{\mathrm{MRL}}$ and $\mathcal{L}_{\mathrm{BPR}}$ exhibit an identical structure.
    \item DCCF achieves superior performance under large dimension sizes for evaluation, but its performance suffers when evaluated with smaller dimension sizes. This observation further validates that the rigid structure imposed by disentangled representation learning could lead to a skewed understanding of user preferences or item features, resulting in biased representations.
\end{itemize}

\begin{table}[t]
\centering
\caption{Performances of MRL-M under different numbers of hierarchical levels.}
\begin{tabular}{@{}cc*{4}{>{\centering\arraybackslash}p{1cm}}@{}}
\toprule
\multirow{2}{*}{\textbf{Dataset}} & \multirow{2}{*}{\textbf{Metric}} & \multicolumn{4}{c}{\textbf{Hierarchical Level}} \\ \cmidrule(l){3-6} & & $L=2$ & $L=3$ & $L=4$ & $L=5$ \\ \midrule
\multirow{6}{*}{\begin{tabular}[c]{@{}c@{}}Amazon-\\ Phones\end{tabular}}
& $Recall@10$ & 0.0876 & 0.0887 & 0.0884 & 0.0906 \\
& $Recall@20$ & 0.1233 & 0.1271 & 0.1263 & 0.1286 \\ 
& $Recall@50$ & 0.1849 & 0.1909 & 0.1886 & 0.1920 \\ \cmidrule(l){2-6}
& $NDCG@10$   & 0.0576 & 0.0586 & 0.0582 & 0.0601 \\ 
& $NDCG@20$   & 0.0680 & 0.0698 & 0.0692 & 0.0710 \\
& $NDCG@50$   & 0.0824 & 0.0846 & 0.0837 & 0.0858 \\ \midrule
\multirow{6}{*}{\begin{tabular}[c]{@{}c@{}}Amazon-\\ Sports\end{tabular}}
& $Recall@10$ & 0.0565 & 0.0561 & 0.0557 & 0.0582 \\
& $Recall@20$ & 0.0816 & 0.0818 & 0.0822 & 0.0845 \\
& $Recall@50$ & 0.1280 & 0.1301 & 0.1292 & 0.1318 \\ \cmidrule(l){2-6}
& $NDCG@10$   & 0.0375 & 0.0375 & 0.0377 & 0.0388 \\ 
& $NDCG@20$   & 0.0451 & 0.0452 & 0.0457 & 0.0469 \\
& $NDCG@50$   & 0.0563 & 0.0569 & 0.0569 & 0.0581 \\ \midrule
\multirow{6}{*}{\begin{tabular}[c]{@{}c@{}}Amazon-\\ Tools\end{tabular}}
& $Recall@10$ & 0.0492 & 0.0504 & 0.0505 & 0.0509 \\
& $Recall@20$ & 0.0692 & 0.0717 & 0.0719 & 0.0734 \\
& $Recall@50$ & 0.1045 & 0.1084 & 0.1073 & 0.1115 \\ \cmidrule(l){2-6}
& $NDCG@10$   & 0.0330 & 0.0330 & 0.0340 & 0.0339 \\ 
& $NDCG@20$   & 0.0390 & 0.0394 & 0.0403 & 0.0404 \\
& $NDCG@50$   & 0.0474 & 0.0481 & 0.0488 & 0.0491 \\ \bottomrule
\end{tabular}
\label{tab:hl}
\end{table}

\subsection{Parameter Sensitivity (RQ4)}
Finally, we study how the number of hierarchical levels in matryoshka representations affects the recommendation performance. Table~\ref{tab:hl} presents the performance results of MRL-M under different number of hierarchical levels. Specifically, $L=2$ means that we set $\mathcal{D}$ as $\{32, 64\}$, $L=3$ means that we set $\mathcal{D}$ as $\{16, 32, 64\}$, and so forth. From the table, we can observe that compared to $L=1$ (BPR-D in Table~\ref{tab:ab}), the introduction of multiple hierarchical levels in MRL-M can indeed enhance performance. However, as $L$ exceeds 3, the incremental performance improvements achieved by further increasing the number of hierarchical levels become limited.

\section{Related Work}
\textbf{Entangled Representation Learning in Recommendation.}
As the most prevalent representation learning technique, entangled representation learning treats each user preference or item feature uniformly, and diffuses its influence indiscriminately across the entire representation~\cite{WWG22,GCP22,PYL23}. For example, BPR~\cite{RFG09} aims to maximize the discrepancy between the predicted ratings for positive and negative items, and BCE~\cite{HKV08} aims to minimize the errors in predicting binary user-item interactions, to refine user and item representations. UIB~\cite{ZZY22} combines BPR and BCE with learned explicit user interest boundaries as bridges. DirectAU~\cite{WYM22} directly optimizes the alignment and uniformity properties of user and item representations. However, such uniform treatment of entangled representation learning may result in the over-generalization problem, where the primary preference or feature overshadows others.

\noindent \textbf{Disentangled Representation Learning in Recommendation.}
Disentangled representation learning~\cite{CCN23,RXZ23,CCW21,ZWZ22,ZZH20,LGL20} has gained popularity in recent years due to its capacity to alleviate the over-generalization problem by separating user preferences and item features into distinct and non-overlapping clusters. For instance, MacridVAE~\cite{MZC19} introduces variational autoencoders to disentangle user intent at different levels based on user behavior data. DGCF~\cite{WJZ20} performs disentangled representation learning via intent-aware interaction graphs to discover user-item relationships and capture fine-grained user intents. CLSR~\cite{ZGC22} disentangles dynamic long- and short-term user interests with contrastive learning in the sequential recommendation. GDCF~\cite{ZLX22} disentangles the latent intent factors across multiple geometric spaces beyond simple Euclidean space. DICE~\cite{ZGL21} learns different representations with disentanglement between interest and conformity based on the structural causal model. However, current disentangled representation learning approaches encounter challenges due to unrealistic constraints imposed by non-overlapping user preferences and item features. Moreover, their rigid structure often results in an incomplete and biased understanding of the data.

\noindent \textbf{Hard Negative Sampling in Collaborative Filtering.}
Negative sampling is essential for implicit collaborative filtering to provide appropriate negative training signals for representation learning. Contrasting with static negative sampling that primarily aims at identifying effective sampling distributions~\cite{RFG09,CSS17,WVS19}, hard negative sampling methods emphasize the significance of oversampling hard negative items to speed up the training process and find more precise delineations of user preferences~\cite{WYZ17,PC19,ZZH22,SCF23}. For example, as the simplest and most prevalent hard negative sampling method, DNS~\cite{ZCW13} proposes to select uninteracted items with higher prediction scores. SRNS~\cite{DQY20} chooses items with both high prediction scores and high variances as negative items to tackle the false negative problem. AHNS~\cite{LCH24} further proposes to adaptively select negative items with different hardness levels to alleviate both the false positive problem and the false negative problem. DENS~\cite{LCZ23} points out the importance of disentangling item factors in negative sampling and designs a factor-aware sampling strategy to identify the optimal negative item for training. Instead of directly selecting negative samples from uninteracted items, MixGCF~\cite{HDD21} synthesizes harder negative items by mixing information from positive items while ANS~\cite{ZCL23} generates more informative negative items from a fine-granular factor perspective.

\section{Conclusion}
In this paper, we introduce a novel representation learning method MRL4Rec, which aims to capture hierarchical user preferences and item features within matryoshka representations. We first represent user preferences and item features at hierarchical levels in matryoshka representation spaces. We theoretically demonstrate that the construction of training triplets tailored to each hierarchical level is crucial in ensuring precise matryoshka representation learning. Based on the theoretical analyses, we further design a matryoshka negative sampling mechanism. Comprehensive experiments confirm that our method provides a promising new research direction for representation learning to further boost the performance of deep-neural-network-based recommender systems.

\clearpage
\balance
\bibliographystyle{ACM-Reference-Format}
\bibliography{RecSys24}

\end{document}